\documentclass[a4paper,USenglish,cleveref,thm-restate,table,numberwithinsect]{lipics-v2021}

\usepackage{tikz}
\usetikzlibrary{math, positioning, decorations.pathreplacing, backgrounds, decorations.shapes}
\usepackage{ifthen}
\usepackage{booktabs}
\usepackage{xspace}
\usepackage{todonotes}

\newcommand{\sep}{\textrm{sep}\xspace}
\newcommand{\eps}{\varepsilon}
\newcommand{\Aa}{\mathcal{A}}
\newcommand{\Bb}{\mathcal{B}}
\newcommand{\Ff}{\mathcal{F}}
\newcommand{\NN}{\mathbb{N}}
\newcommand{\PP}{\mathbb{P}}
\newcommand{\ZZ}{\mathbb{Z}}

\DeclareMathOperator{\poly}{poly}
\DeclareMathOperator{\polylog}{polylog}

\newcommand{\Prob}[2][]{%
  \ifthenelse{\equal{#1}{}}%
  {\PP}%
  {\PP_{ #1 }}%
  \left[ #2 \right]
}

\newcommand{\wordrect}[5]{
  \draw[black] (#4 + 0,0) rectangle ++(#1 ,.5);
  \coordinate (#5) at (#4 + #1/2, 0);
  \node (a) at (#4 + 0.6, 0.25) { #3 };
  \draw[black, |<->|] (#4 + 0, 0.6) -- ++(#1 ,0) node[midway, above] { #2 };
}

\theoremstyle{plain}
\newtheorem{fact}[theorem]{Fact}

\title{Separating words with automata in the half-adversarial case}
\author{Gabriel Bathie}{LaBRI, Université de Bordeaux}{gabriel.bathie@gmail.com}{https://orcid.org/0000-0003-2400-4914}{This work was supported by a Postdoctoral Fellowship of the PEPR IA SAIF.}
\date{}

\authorrunning{G. Bathie}
\Copyright{Gabriel Bathie}

\hideLIPIcs
\nolinenumbers

\begin{document}
\maketitle

\begin{abstract}
  We consider the problem of separating words with deterministic finite automata (DFA) (Goral{\v{c}}{\'i}k and Koubek, 1986).
  This problem asks: given two distinct words $u,v$ of length at most $n$, what is the size of the smallest DFA that accepts one and rejects the other?
  The best upper bound on the worst-case over all pairs of words of length at most $n$ is $\tilde{O}(n^{1/3})$ states (Chase, 2021), while the best lower bound is $\Omega(\log n)$.

  In this work, we consider the half-random, half-adversarial case: we show that if $u$ is a uniformly random binary word of length $n$, then with high probability, for any word $v$ not equal to $u$, there is a DFA with $O(\log^{7/3} n  \poly\log\log n)$ states that separates $u$ and $v$.
  Our results are based on a novel analysis that exploits the structural sparsity of random words: we show how to apply block-wise compaction with small deterministic transducers to reduce the separation problem to the case of words with short run-length encodings.
\end{abstract}

\section{Introduction}
We consider the problem of separating words with deterministic finite automata (DFA), introduced by Goral{\v{c}}{\'i}k and Koubek~\cite{goralcik1986discerning} in 1986.
This problem asks: given two distinct words $u,v$ of length at most $n$, what is the size (number of states) of the smallest DFA that \emph{separates} them, i.e., accepts one and rejects the other?

A solution to this problem is a tight asymptotic worst-case estimation of the minimal number of states, as a function of the length of $u$ and $v$.
This problem, while easy to state, has proven rather hard to solve: finding a solution has been an open problem for more than 40 years, and there is still a large gap between the state-of-the-art upper and lower bounds.
The state of the art is the following: in 2020, Chase~\cite{chase2020new,chase2021separating} showed that any two distinct words of length at most $n$ can be separated by a DFA with $\tilde{O}(n^{1/3})$ states. For the lower bounds, it is known that there exist words of length $n$ for which $\Omega(\log n)$ states are required~\cite{goralcik1986discerning,demaine2011remarks}.

No better bounds are known in the general case, but several restricted cases have been studied. For example, it is easy to show that if $u$ and $v$ are uniformly random words of length $n$ over a binary alphabet, then with high probability, there exists a DFA with $O(\log n)$ states that separates them~\cite{demaine2011remarks}.

\subsection{Our results}

In this work, we consider the half-random, half-adversarial case, as an intermediate step between the exact problem and the random case, and show that a polylogarithmic number of states is sufficient for this case.
We show that if $u$ is a uniformly random binary word of length $n$, then with high probability, for any word $v$ not equal to $u$, there is a DFA with $O(\log^{7/3} n  \poly\log\log n)$ states that separates $u$ and $v$.

\begin{restatable}{theorem}{mainresult}\label{thm:main}
  There exists a constant $C$ such that for every $n$, if $u$ is a word of length $n$ chosen uniformly at random over $\{a,b\}^n$, we have:
  \[\Prob{\forall v \ne u, \sep(u, v) = O(\log^{7/3} n \poly\log\log n)} \ge 1 - \frac{1}{n^C}.\]
\end{restatable}

Our results are based on a novel analysis that exploits the structural sparsity of random words: we show how to apply block-wise compaction with small deterministic transducers to reduce the separation problem to the case of words with short run-length encodings.

\subsection{Related work}

\subparagraph*{Upper bounds, lower bound, and history.}
The separating word problem was first introduced by Goral{\v{c}}{\'i}k and Koubek~\cite{goralcik1986discerning} in 1986, who showed that $o(n)$ states were always sufficient.
This upper bound was later improved by Robson~\cite{robson1989separating} to $O(n^{2/5} \polylog n)$ states in 1989.
No further improvement was found until Chase~\cite{chase2020new} proved in 2020 that $O(n^{1/3} \log^7 n)$ states suffice.

The lower bound of $\Omega(\log n)$ states was already given by Goral{\v{c}}{\'i}k and Koubek~\cite{goralcik1986discerning}, and no further improvement is known in the general case.

It has been conjectured that the asymptotic optimal behavior is logarithmic, i.e., that any pair of distinct words can be separated by a DFA with $O(\log n)$ states. Yakaryilmaz and Montoya~\cite{yakaryilmaz2015discerning} gave empirical evidence supporting this conjecture.

\subparagraph*{Special cases.}
Demaine et al.~\cite{demaine2011remarks} show upper and lower bounds for multiple special cases.
For example, they remark that if the words differ within their first $d$ or last $d$ positions, then they can be separated by automata with $d+2$ states.
They also show that if the Hamming distance between $u$ and $v$ is at most $d$, then they can be separated by an automaton with $O(d\log n)$ states.
This result is somewhat counter-intuitive: the hard-to-separate words are \emph{not} words that are very similar.

\subparagraph*{Variants.}
Several variants of this problem have been considered, such as separation by other kinds of computational devices, e.g., nondeterministic finite automata \cite{demaine2011remarks}, quantum finite automata~\cite{belovs2016can}, context-free grammars~\cite{currie1999separating}, finite groups~\cite{robson1996separating,kuntewar2023separating} and others~\cite{belovs2017conjecture}.
Other variants include separation from \emph{any} starting state ($\forall$-separation) and separation from any \emph{pair} of starting states (one for each word, $\forall^2$-separation)~\cite{tran2022variations,tran2023separating}.

\section{Technical preliminaries}
Throughout this work, $\log$ denotes the logarithm in base $2$.

\subsection{Words, finite automata, and transducers}
\paragraph*{Words}
Let $\Sigma$ be a finite set, which we call the alphabet.
A word over $\Sigma$ is a finite ordered sequence of elements of $\Sigma$.
For $n\ge 0$, the set of words of length $n$ over $\Sigma$ is denoted $\Sigma^n$, and the set of all words over $\Sigma$ is denoted $\Sigma^*$.
Additionally, $\Sigma^{\le n}$ denotes the set of all words of length at most $n$.
The notation $\eps$ denotes the empty word, that is, the unique word of length $0$.
For a word $u \in \Sigma$, $|u|$ denotes its length, and for $0\le i < |u|, u[i]$ denotes its $i$-th element, called its $i$-th \emph{letter}.
Given two words $u$ and $v$, $u\cdot v$ denotes the \emph{concatenation} of $u$ and~$v$. For shorthand, we may write $uv$ for $u\cdot v$.
For $0\le i < j <n$, $u[i:j]$ denotes the word $u[i]\cdot u[i+1] \cdot \ldots \cdot u[j-1]$.
Given a non-negative integer $\alpha$ and a word $u$, $u^\alpha$ denotes the concatenation of $\alpha$ copies of $u$.
We say that a word $v$ is a prefix of $u$ if there exists an index~$i$ such that $v = u[0:i]$.

\paragraph*{Finite automata}
We recall here the notion of deterministic finite automata.
\begin{definition}[Deterministic Finite Automaton]
  A \emph{deterministic finite automaton} (DFA) is a tuple $\Aa =(\Sigma, Q, \delta, q_0, F)$ where:
  \begin{itemize}
    \item $\Sigma$ is a finite alphabet,
    \item $Q$ is a finite set of states, $q_0\in Q$ is the \emph{initial} state, and $F \subseteq Q$ is the set of \emph{accepting} states,
    \item $\delta : Q\times \Sigma \to Q$ is the \emph{transition} function.
  \end{itemize}
  The function $\delta$ inductively extends into a function $\delta^*:Q\times \Sigma^* \to Q$ defined as:
  \begin{align*}
    \forall q\in Q,~&
    \delta^*(q, \eps) = q\\
    \forall q\in Q, \forall a\in \Sigma,~&
    \delta^*(q, a\cdot w) = \delta^*(\delta(q, a), w).
  \end{align*}
  We use $p\xrightarrow{u}q$ as a shorthand to denote the fact that $q = \delta^*(p, u)$, and in this case we say that ``$u$ labels a run from $p$ to $q$ in $\Aa$''.

  Finally, we say that $\Aa$ \emph{accepts} $u$ if $u$ labels a run from the initial state $q_0$ to some accepting state, i.e., if $\delta^*(q_0, u) \in F$.
  The \emph{language recognized by $\Aa$}, denoted L$L(\Aa)$, is the set of words accepted by $\Aa$: $L(\Aa) = \{u: \delta^*(q_0, u)\in F\}$.
\end{definition}
We use $|\Aa|$ to denote the number of states of $\Aa$ (that is, the cardinality of $Q$).

Given two words $u,v$, we say that they are \emph{separated} by $\Aa$ if exactly one of them is accepted by $\Aa$,
or equivalently: $\textrm{card}(\{u,v\} \cap L(\Aa)) = 1$.
The function $\sep: \Sigma^*\times\Sigma^* \to \NN$ maps a pair of words to the number of states of the smallest automaton that separates them:
\[\forall u,v \in \Sigma, \sep(u,v) = \min_{\Aa: \Aa \text{ separates } u, v} |\Aa|.\]
For a non-negative integer $n$, we use $\sep(n)$ to denote the maximum of $\sep(u,v)$ for $u$ and $v$ of length at most $n$.
The Separating Words problem aims to find the asymptotic behavior of~$\sep(n)$: the upper bounds of Chase~\cite{chase2021separating} shows that $\sep(n) = O(n^{1/3} \log^7 n)$, and the best lower bound shows $\sep(n) = \Omega(\log n)$.

\paragraph*{Transducers}
In the proof of \cref{thm:main}, we will apply successive transformations to the words until they become easy to separate.
We want the final computation to be implemented by a finite automaton; hence we want these transformations to be ``computable by a DFA''. This property is formally defined using the notion of \emph{deterministic finite transducers} (DFT)\footnote{In this work, all the transducers that we consider are deterministic and finite, therefore we simply use ``transducers'' to refer to DFTs.}, an extension of deterministic finite automata with outputs.
This is the same notion as \emph{translators} used in~\cite{goralcik1986discerning}.

\begin{definition}[Transducer]
  A \emph{transducer} $T$ is a tuple $(\Sigma_1, \Sigma_2, Q, \delta, \mu, q_0)$ where:
  \begin{itemize}
    \item $\Sigma_1$ and $\Sigma_2$ are finite alphabets,
    \item $Q$ is a finite set of states, $q_0\in Q$ is the \emph{initial} state, and
    \item $\delta : Q\times \Sigma_1 \to Q$ is the \emph{state transition} function and $\mu : Q\times \Sigma_1 \to (\Sigma_2 \cup \{\eps\})$ is the \emph{output} function\footnote{In the usual definition of deterministic transducers, the functions $\delta$ and $\mu$ are presented as a single function $Q\times \Sigma_1  \to Q\times (\Sigma_2^* \cup \{\eps\})$. We choose to separate them to make subsequent definitions easier to state and manipulate.}.
  \end{itemize}
  Similarly to finite automata, $\delta$ and $\mu$ inductively extend into functions $\delta^*:Q\times \Sigma_1^* \to Q$ and $\mu^*:Q\times \Sigma_1^* \to \Sigma_2^*$. The latter is defined as follows:
  \begin{align*}
    \forall q\in Q,~&
    \mu^*(q, \eps) = \eps\\
    \forall q\in Q, \forall a\in \Sigma_1,~&
    \mu^*(q, a\cdot w) = \mu(q, a)\cdot \mu^*(\delta(q, a), w).
  \end{align*}
  Finally, $T$ defines a function $T(\cdot): \Sigma_1^*\to\Sigma_2^*$ given by
  $T(u) = \mu^*(q_0, u)$.
\end{definition}

Next, we show that transducers and automata compose as expected.

\begin{lemma}[Transducer-automaton composition]\label{lemma:compose-transducer}
  Let $T= (\Sigma_1, \Sigma_2, Q_T, \delta_T, \mu_T, q_T)$ be a transducer with $t$ states and let $\Aa = (\Sigma_2, Q_\Aa, \delta_\Aa, q_\Aa, F_\Aa)$ be a finite automaton with~$s$ states.
  There exists an automaton $\Aa_T$ with $s\cdot t$ states that recognizes the language $L = \{u : T(u) \in L(\Aa)\}$.
\end{lemma}
\begin{proof}
  Informally, $\Aa_T$ is a product automaton of $T$ and $\Aa$: it tracks the states of both automata, and when $T$ outputs a letter, it follows the corresponding transition in $\Aa$.

  Formally, we let $\Aa_T = (\Sigma_1, Q, \delta, q_0, F)$, where:
  \begin{itemize}
    \item $Q = Q_T \times Q_\Aa$ and $q_0 = (q_T, q_\Aa)$,
    \item $F = Q_T \times F_\Aa$,
    \item For any $q_t\in Q_T, q_a\in Q_\Aa$ and $x\in\Sigma_1$, $\delta((q_t, q_a), x) = (\delta_T(q_t, x), \delta_\Aa(q_a, \mu(q_t, x)))$, where we slightly abuse notation and define $\delta_\Aa(q_a, \eps) = \eps$.
  \end{itemize}
  One can easily show by induction on $w$ that $\delta^*((q_t, q_a), w) = (\delta_T^*(q_t, w), \delta_\Aa^*(q_a, \mu^*(q_t, w)))$, and it follows that $\Aa_T$ accepts $u$ if and only if $\Aa$ accepts $T(u)$.
\end{proof}

We conclude this section with an observation that follows from \cref{lemma:compose-transducer}: if applying a small transducer to $u$ and $v$ results in easy-to-separate words, then $u$ and $v$ are also easy to separate.

\begin{corollary}[{\cite[Lemma 2]{goralcik1986discerning}}]\label{cor:compose}
  Let $T$ be a transducer with $t$ states, and let $u,v$ be words such that $T(u)$ and $T(v)$ can be separated by an automaton of size $s$.
  Then $u$ and $v$ can be separated by an automaton of size $s\cdot t$.
\end{corollary}

\subsection{Simplifying assumptions}\label{sec:simplifying}
When seeking upper bounds, we can make a number of simplifying assumptions about $u$ and~$v$.
First, we can assume that they both have length $n$, as otherwise they can be separated by an automaton with $O(\log n)$ states, which matches the lower bound.
This observation relies on the following fact.
\begin{fact}[{\cite[Example 6]{SHALLIT199610}}]\label{fact:modp}
  Let $i\ne j$ be integers such that $0\le i,j \le n$ and $n\ge 2$.
  Then there exists a prime number $p = O(\log n)$ such that $i\ne j \pmod{p}$.
\end{fact}
If $u$ and $v$ have different lengths bounded by $n$, then by \cref{fact:modp} there exists a prime $p$ such that $|u|\ne|v| \pmod{p}$. Consider the automaton that counts the number of letters in the input word, with one state per residue modulo $p$, and only accepts in the state corresponding to $|u| \pmod{p}$; this automaton accepts $u$ and not $v$, and has $p = O(\log n)$ states.

For similar reasons, we can assume that $u$ and $v$ contain the same number of occurrences of $a$, for each letter $a\in \Sigma$, or even the same number of occurrences of each word $w$ of constant length: otherwise, they could be separated by an automaton that counts occurrences of~$a$ or~$w$ modulo a suitable $p$.

Furthermore, we can assume w.l.o.g. that $u$ and $v$ are binary words.
Indeed, assume that~$u$ and~$v$ are defined over an alphabet $\Sigma$ of size larger than two, and let $\sigma : \Sigma \rightarrow \{a,b\}$ define a morphism over $\Sigma^*$ such that $u' = \sigma(u) \ne \sigma(v) = v'$.
Then, the number of states required to separate $u'$ and $v'$ is at least the number required for $u$ and~$v$: given an automaton~$\Aa$ that separates $u'$ and $v'$, we can build an automaton of the same size that separates $u$ and~$v$ by, given a letter $a$, following the transition labeled by $\sigma(a)$ in $\Aa$.

More formally, if $\Aa = (Q, \{a,b\}, \delta_\Aa, q_0, F)$ is a DFA with $t$ states that separates $u'$ and~$v'$, we can build the automaton $\Bb = (Q, \Sigma, \delta_\Bb, q_0, F)$ that separates $u$ and $v$ by defining $\delta_\Bb(q, a) = \delta_\Aa(q, \sigma(a))$. It follows that
\begin{align*}
  \delta_\Bb(q_0, u) &= \delta_\Aa(q_0, \sigma(u)) = \delta_\Aa(q_0, u') \text { and}\\
  \delta_\Bb(q_0, v) &= \delta_\Aa(q_0, \sigma(v)) = \delta_\Aa(q_0, v').
\end{align*}
By construction, $\Aa$ and $\Bb$ have the same number of states, and since $\Aa$ separates $u'$ and $v'$, hence $\Bb$ separates $u$ and~$v$.
It follows that for any such morphism $\sigma$, we have \[sep(u,v) \le sep(\sigma(u), \sigma(v)).\]

The proof of our main result uses a technique that extends the one used in the above proof: the main difference is that the morphism is applied to blocks of letters instead of a single letter.

More remarks and simplifying assumptions about the problem of separating words can be found in~\cite{demaine2011remarks}.

\section{Proof overview}\label{sec:overview}
In this section, we give a high-level overview of the construction behind our main result.
The formal proof is given in \cref{sec:fullproof}.
\mainresult*

\subparagraph*{Proof idea.}
We start by dividing~$u$ and~$v$ into blocks of length $k = O(\log n)$.
This gives us new words $U, V$ over a larger alphabet $\Sigma = \{a,b\}^k$.
Since $u\ne v$, we also have $U \ne V$, and there exists a position~$i$ such that $U[i] \ne V[i]$. Let $B = U[i]$.
Since~$u$ is random,~$U$ contains at most $O(1)$ copies of each block of length~$k$, in expectation.
Therefore, with probability at least $1- {1}/{n^C}$, $U$ contains at most $O(\log n)$ occurrences of any block, and in particular, $U$ contains at most $O(\log n)$ occurrences of $B$.
Consider the morphism $\sigma : \Sigma \rightarrow \{a,b\}$ defined by:
\begin{align*}
  \sigma(B) &= a\\
  \sigma(B') &= b, \text{ for all } B' \ne B.
\end{align*}
Let $u' = \sigma(U)$ and $v' = \sigma(V)$.
By construction, $u'$ contains $O(\log n)$ occurrences of $a$ (w.h.p.).
Notice that applying the morphism $\sigma$ can be done using a ``transducing'' DFA with $2k = O(\log n)$ states.

Now, $u'$ and $v'$ are of the form
\[
  u' = a^{\alpha_1}b^{\alpha_2}\ldots b^{\alpha_t}, \text{ and }
  v' = a^{\alpha_1'}b^{\alpha_2'}\ldots b^{\alpha_{t'}'},
\]
If $t\ne t'$, we can separate them using $O(\log n)$ states by counting the number of occurrences of $ab$ modulo a small prime $p = O(\log n)$. By composing with the above construction, this gives $O(\log^2 n)$ states.

If $t = t'$, there must exist $i$ such that $\alpha_i \ne \alpha_i'$. Since $\alpha_i, \alpha_i' \le n$, there is a prime $p = O(\log n)$ such that $\alpha_i \ne \alpha_i' \pmod{p}$.
We use another transducing automaton of size $O(p)$ that counts the number of successive $a$'s mod $p$, and upon seeing $ab$, outputs $0$ if the counter is equal to $\alpha_i \pmod{p}$, and $1$ otherwise. (Note: we might have to do this with $b$'s instead if~$i$ is even.)
The resulting words $u''$ and $v''$ are different, and have length $t = O(\log n)$ (since there are $O(\log n)$ $a$'s in $u'$).
Using the construction of Chase~\cite{chase2021separating}, we can separate them using $O(\log^{1/3} n \poly\log\log n)$ states.
Combining all of the above gives us $O(\log^{7/3} n \poly\log\log n)$ states.

\section{The Half-Random, Half-Adversarial Case}\label{sec:fullproof}
In this section, we give a formal proof of our main result.
We will work backwards compared to \cref{sec:overview}, with each step separated as a standalone lemma with the necessary assumption; the main result will follow by combining them.

\subsection{Separating sparse words}\label{sec:sparse}
We first show that if a word $u$ contains a small number of occurrences of any given letter, then it is easy to separate from all other words, since $u$ has a short \emph{run-length encoding}.
We call such a word a \emph{sparse} word.
\begin{definition}
  A word over $\Sigma = \{a,b\}$ is \emph{$t$-sparse} if it contains at most $t$ occurrences of~$a$.
\end{definition}

\begin{lemma}\label{lemma:sparse}
  Let $\Sigma = \{a,b\}$, and let $u$ be a $t$-sparse word of length $n$.
  Then, for any other word $v \in \Sigma^{\le n}$, $u$ and $v$ can be separated by an automaton with $O(\log n \cdot t^{1/3} \polylog t)$ states.
\end{lemma}
\begin{proof}
  We show that $u$ and $v$ can be transduced to distinct words $u'$ and $v'$ of length $O(t)$ using a DFT with $O(\log n)$ states. Using the result of Chase~\cite{chase2021separating}, $u'$ and $v'$ can be separated by an automaton of size $O(t^{1/3} \poly\log t)$. Composing the automaton and the transducer yields the desired result.

  We consider the \emph{run-length encoding} of $u$ and $v$.
  There exist integers $\ell,\ell'$ and $(\alpha_i)_{i=0,\ldots, \ell}$, $(\alpha_i')_{i=0,\ldots, \ell'}$, such that:
  \begin{itemize}
    \item $u = a^{\alpha_1}b^{\alpha_2}\ldots b^{\alpha_\ell}$,
    \item $v = a^{\alpha_1'}b^{\alpha_2'}\ldots b^{\alpha_{\ell'}'}$,
    \item all $\alpha_i$ and $\alpha_i'$ are strictly greater than $0$, except perhaps the first and last one of each sequence.
  \end{itemize}
  Since $u$ is $t$-sparse, we have $\ell \le 2 t$.
  Next, we consider three cases and show that in each of them, $u$ and $v$ can be separated by a DFA with at most $O(\log n \cdot t^{1/3} \polylog t)$ states.

  \subparagraph*{Case 1: $\ell \ne \ell'$.}
  Then, $u$ and $v$ have a different number of occurrences of $ab$ (or $ba$), and can be separated by an automaton with $O(\log n)$ states that counts the number of occurrences of this substring, using \cref{fact:modp}, as described in \cref{sec:simplifying}.

  \subparagraph*{Case 2: $\ell = \ell', \alpha_\ell \ne \alpha_\ell'$.}
  If one of $\alpha_\ell$ or $\alpha_\ell'$ is $0$, the other is not; hence exactly one of $u$ and $v$ has a $b$ as its last letter, and they can be separated by an automaton with two states.
  Otherwise, if both $\alpha_\ell$ and $\alpha_\ell'$ are non-zero, $u$ and $v$ can be separated by an automaton with $O(\log n)$ states that counts the number of trailing $b$'s modulo some small $p = O(\log n)$, once again using \cref{fact:modp}.

  \subparagraph*{Case 3: $\ell = \ell', \alpha_\ell = \alpha_\ell'$.}
  Then, as $u\ne v$, there exists an index $i < \ell$ such that $\alpha_i \ne \alpha_i'$.
  Furthermore, since $u$ and $v$ have length at most $n$, we have $\alpha_i, \alpha_i' \le n$.
  Therefore, by \cref{fact:modp}, there exists a prime $p = O(\log n)$ such that $\alpha_i \ne \alpha_i' \pmod{p}$.
  We construct a transducer $T: \Sigma \to \{c, d\}$ that counts modulo $p$ the length $L$ of the longest suffix that consists only of~$a$'s or only of $b$'s at the current position; upon seeing the other letter, outputs $c$ if $L = \alpha_i \pmod{p}$ and $d$ otherwise.
  This construction requires $\alpha_\ell = \alpha_\ell'$: since the transducer only outputs a value upon seeing the other letter, it does not output anything for the last segment.

  Let $u' = T(u)$ and $v' = T(v)$; see \cref{fig:proof} for an illustration of the action of $T$ on~$u$ and $v$.
  By construction, for $j<\ell$, the $j$-th letter of $u'$ (resp. $v'$) is $c$ iff $\alpha_j = \alpha_i \pmod{p}$ (resp. $\alpha_j' = \alpha_i \pmod{p}$). Therefore, $u'\ne v'$, and $u'$ and $v'$ both have length $\ell-1 \le 2t$.
  Consequently, we can use the construction of Chase~\cite{chase2021separating} to separate them using $O(t^{1/3} \poly\log t)$ states.

  \begin{figure}[htbp]
    \centering
    \begin{tikzpicture}
      \node (u) at (-0.3, 0.25) {$u$};
      \wordrect{2}{$\alpha_1$}{$aaa\ldots$}{0}{n1}
      \wordrect{1.5}{$\alpha_2$}{$bbb\ldots$}{2}{n2}
      \draw[black] (3.5,0) rectangle ++(1,.5) node[midway] {$\ldots$};
      \wordrect{2}{$\alpha_{i}$}{$aa\ldots$}{4.5}{n3}
      \draw[black] (6.5,0) rectangle ++(1,.5) node[midway] {$\ldots$};
      \wordrect{2.5}{$\alpha_{\ell-1}$}{$aa\ldots$}{7.5}{n4}
      \wordrect{2}{$\alpha_\ell$}{$bbb\ldots$}{10}{n5}

      \begin{scope}[shift={(3.5,-1.5)}]
        \node (u) at (-0.3, 0.3) {$u'$};
        \draw[black] (0.0,0) rectangle ++(.5,.5) node[inner sep=5pt, midway] (b1) {$c$};
        \draw[black] (0.5,0) rectangle ++(.5,.5) node[inner sep=4pt, midway] (b2) {$d$};
        \draw[black] (1.0,0) rectangle ++(1,.5) node[inner sep=5pt, midway] {$\ldots$};
        \draw[black] (2.0,0) rectangle ++(.5,.5) node[inner sep=5pt, midway] (b3) {$c$};
        \draw[black] (2.5,0) rectangle ++(1.,.5) node[inner sep=5pt, midway] {$\ldots$};
        \draw[black] (3.5,0) rectangle ++(.5,.5) node[inner sep=4pt, midway] (b4) {$d$};
      \end{scope}

      \begin{scope}[shift={(3.5,-2.5)}]
        \node (u) at (-0.3, 0.3) {$v'$};
        \draw[black] (0,0) rectangle ++(2,.5) node[midway] {$\ldots$};
        \draw[black] (2.0,0) rectangle ++(.5,.5) node[midway] {$\color{red} d$};
        \draw[black] (2.5,0) rectangle ++(1.5,.5) node[midway] {$\ldots$};
      \end{scope}

      \draw[thick, blue, ->] (n1) to[out=-75, in=120] (b1.north);
      \draw[->] (n2) to[out=-90, in=90] (b2.north);
      \draw[thick, blue, ->] (n3) to[out=-90, in=90] (b3.north);
      \draw[->] (n4) to[out=-90, in=90] (b4.north);
    \end{tikzpicture}
    \caption{Illustration of the action of the transducer $T$ used in Case 3 for separating words with different run-length encodings of the same length.
    Thick blue arrows correspond to indices $j$ such that $\alpha_j \equiv \alpha_i \pmod{p}$, where $i$ is an index such that $\alpha_i \ne \alpha_i'$, therefore $u'$ and $v'$ differ at the corresponding position.}\label{fig:proof}
  \end{figure}
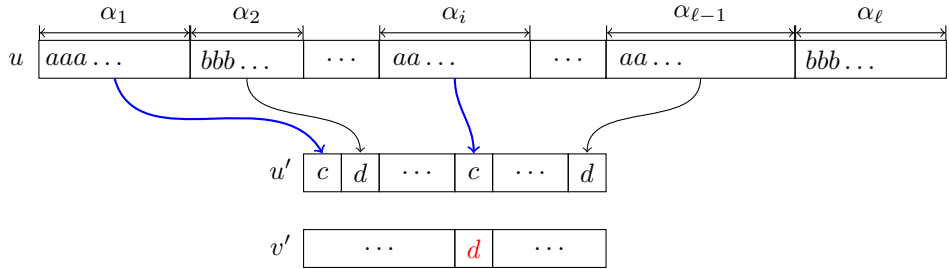

  The transducer $T$ can be implemented with $2p = O(\log n)$ states, using the state set $\{a,b\} \times \ZZ_p$: the first element corresponds to the last letter seen while the second corresponds to the current residue of $L$ mod $p$.
  Therefore, by \cref{cor:compose}, $u$ and $v$ can be separated using $O(\log n \cdot t^{1/3} \polylog t)$ states.
\end{proof}

\subsection{Properties of random words}\label{sec:random}
In this section, we show that there is a family $\Ff$ of transducers of size $O(\log n)$ such that
\begin{enumerate}
  \item with high probability, a random word $u$ is transduced into a $O(\log n)$-sparse word by \emph{every} transducers in $\Ff$, and
  \item for any word $v\ne u$, there exists a transducer $T\in\Ff$ that maps $u$ and $v$ to distinct words.
\end{enumerate}
Composing such a transducer with the result of \cref{sec:sparse} proves \cref{thm:main}.
In this section, we assume that the alphabet is $\Sigma = \{a, b\}$.

Given a word $w$ of length $k$, we define the transducer $T_w$ with the following action:
given an input $u$, it divides it into blocks of length $k$, and maps each block to a single letter: $a$ if the block is equal to~$w$, and $b$ otherwise.
It is formally defined as follows: $T_w = (\Sigma, Q, \delta, \mu, q_0)$, where $Q = {0,\ldots, k-1} \times \{\top, \bot\}$ and $q_0 = (0, \top)$. The state $(i, \top)$ represents the fact that the first $i$ letters of the current block match the prefix of length $i$ of $w$, and $(i, \bot)$ means that it is not the case. Upon reading a block of length $k$, $\mu$ outputs $a$ if the block is equal to~$w$, and $b$ otherwise.
To preserve these properties, $\delta$ and $\mu$ are defined as:
\begin{align*}
  \forall i < k-1,~\delta((i, \top), x) &=
  \begin{cases}
    (i+1, \top)& \text{ if } x = w[i] \\
    (i+1, \bot)& \text{ otherwise}
  \end{cases}\\
  \forall i < k-1,~\delta((i, \bot), x) &= (i+1, \bot)\\
  \delta((k-1, \bot), x) &= \delta((k-1, \top), x) = (0, \top)\\
  \forall i < k-1,~\mu((i, \bot), x) &= \mu((i, \top), x) = \eps\\
  \mu((k-1, \bot), x)&= b\\
  \mu((k-1, \top), x)&=
  \begin{cases}
    a& \text{ if } x = w[k-1] \\
    b& \text{ otherwise}
  \end{cases}
\end{align*}
This transducer has $2k-1$ states (as $(0, \bot)$ is unreachable).
Finally, we define the family $\Ff_k$ as \[\Ff_k = \{T_w : w\in\Sigma^k\}.\]
Next, we turn to proving that it has the desired properties.

\begin{lemma}\label{lemma:random}
  Let $n$ be a large enough integer, and let $C > 0$ be a constant.
  Let $u$ be chosen uniformly at random over $\Sigma^n$, and let $k = \log n$ and let $t = (C+1) \log n$.
  Then, with probability at least $1 - 1/n^C$ (over the choice of $u$), $T(u)$ is a $t$-sparse word for all $T\in\Ff_k$.
\end{lemma}
\begin{proof}
  By the definition of $\Ff_k$, it suffices to show that, with probability at least $1 - 1/n^C$, the word $u$ contains at most $t$ occurrences of any word $w$ of length $k$.
  This follows from a standard concentration-of-measure argument (Chernoff's bound) and the Union Bound.

  Let $B_i$ denote the $i$-th block of length $k$ in $u$; there are $m = n/k$ such blocks in $u$. Since~$u$ is uniformly random, $B_i$ is independent of $B_j$ for any $i\ne j$.
  Fix an arbitrary word $w$ of length $k$.
  Let $X_i$ denote the random variable whose value is $1$ if $B_i$ is equal to $w$, and $0$ otherwise.
  By construction, the probability that $B_i$ is equal to $w$ is $p = 2^{-k}$ (recall that, in this section, we consider a binary alphabet).
  By Chernoff's bound (\cref{thm:chb}), for any $R \ge 6 m \cdot p$, we have:
  \[\Prob{\sum_i X_i > R} \le 2^{-R}.\]
  Plugging in $k = \log n$, $p = 1/n$, $m = n/\log n$ and $R = (C+1) \log n$, we get:
  \begin{align*}
    \Prob{\sum_i X_i > (C+1)\log n}
    &\le 2^{-(C+1)\log n} = \frac{1}{n^{C+1}}.
  \end{align*}
  This holds as soon as $R \ge 6 m \cdot p$, which is equivalent to $\log^2 n \ge \frac{6}{(C+1)}$,
  which is true for $n$ larger than $2^{\sqrt{6/(C+1)}}$.

  It follows that $u$ contains more than $t = (C+1) \log n$ blocks equal to $w$ with probability at most $\frac{1}{n^{C+1}}$.
  Since $k=\log n$, there are $n$ distinct possible values of $w$.
  By the union bound, the probability that there exists a $w$ such that $u$ contains more than $t$ blocks equal to $w$ is~$\frac{1}{n^{C}}$.

  A transducer $T = T_w\in\Ff_k$ maps each block equal to $w$ to the letter $a$, and the others to $b$.
  Therefore, $T(u)$ is $t$-sparse whenever $u$ contains at most $t$ blocks equal to $w$, and this holds simultaneously for all $w$ with probability at least $\frac{1}{n^{C}}$.
\end{proof}

\begin{lemma}\label{lemma:distinct}
  Let $u,v$ be distinct words of length $n$. Then for any $k$, there exists $T\in\Ff_k$ such that
  $T(u) \ne T(v)$.
\end{lemma}
\begin{proof}
  It suffices to choose $T = T_w$ where $w$ is a block of $u$ that is aligned with a different block in~$v$.

  More formally, as $u\ne v$, there exists an index $i$ such that $u[i] \ne v[i]$. Let $l = k\cdot \lfloor i/k\rfloor$ be the largest multiple of $k$ less than or equal to $i$, and let $r = l + k$.
  For $w = u[l:r]$ and $T = T_w$, we have that the $i$-th letter of $T(u)$ is $a$, while the corresponding letter of $T(v)$ is $b$, hence $T(u) \ne T(v)$.
\end{proof}

\subsection{Separating a random word from an adversarial word}\label{sec:mainproof}
We now turn to proving \cref{thm:main}.
\begin{proof}[Proof of \cref{thm:main}]
  If $u$ is chosen at random, \cref{lemma:random} implies that, with probability at least $1-1/n^C$, any transducer in $\Ff_{\log n}$ maps $u$ to an $O(\log n)$-sparse word.
  This is in particular true for the transducer given by \cref{lemma:distinct} such that $T(u) \ne T(v)$.

  Now, according to \cref{lemma:sparse}, if $u'$ is $O(\log n)$-sparse, then $u'$ can be separated from any $v'\ne u'$ using $O(\log^{4/3} n \poly\log\log n)$ states.

  Composing $T$ with this observation applied to $u' =T(u)$ and $v' = T(v)$ yields a DFA with $O(\log^{7/3} n \poly\log\log n)$ states that separates $u$ and $v$.
\end{proof}

\section{Conclusion}
\subsection{Insights on hard-to-separate words}
The results and the techniques presented in this paper provide some insight into the nature of words that are hard to separate.
The first observation is that, for small $k$, if the block $A = u[i:i+k]$ is not equal to $B = v[i:i+k]$, then, both of these blocks must occur many times in $u$ and $v$; otherwise, applying one of the transducers $T_A$ or $T_B$ leads to a sparse word, which is easy to separate from any other word (\cref{lemma:sparse}).
This is in particular true if \emph{any} block of $u$ (and $v$) appears many times in $u$ and $v$.
This is not the case in a typical random word of length $n$: the expected number of occurrences of a block of length $k$ is $n/(k\cdot 2^k)$.
However, this is the case in words that are suspected to be hard cases, such as prefixes $S[0:n]$ of the Thue--Morse sequence $S$: for each $k$, any factor of length $k$ of $S$ occurs $\Theta(n/k)$ times in $S[0:n]$.

Therefore, it is a compelling idea to try to build hard-to-separate pairs by selecting a few blocks of length $k = \Omega(\log n)$, and interleaving them randomly. Such a construction would ensure that each of these blocks occurs $\Theta(n/k)$ times, and that the resulting words have many mismatches (with high probability), thereby also avoiding the construction from the case of a low-Hamming-distance pair.
However, the second observation that this work provides is that the property of having many occurrences of subwords must be true \emph{hierarchically}, in the sense that, if we map each block to a single letter, the resulting word still has the same property.
This is not the case when randomly interleaving words; instead, the construction should be recursive, as is the case for the Thue--Morse sequence.

\subsection{Future work}
The best known lower bound for the general case remains $\Omega(\log n)$.
Our analysis indicates that candidate hard-to-separate pairs must exhibit scale-invariant factor repetitions to resist block-transducer reductions.
A pressing open question is whether there is a hierarchical (recursive, fractal) construction that produces a counterexample to the $O(\log n)$ conjecture.

\bibliographystyle{plainurl}
\bibliography{biblio}

@article{chase2020new,
  title   = {A new upper bound for separating words},
  author  = {Chase, Zachary},
  journal = {arXiv preprint arXiv:2007.12097},
  year    = {2020}
}

@inproceedings{chase2021separating,
  title     = {Separating words and trace reconstruction},
  author    = {Chase, Zachary},
  booktitle = {Proc. of STOC},
  pages     = {21--31},
  year      = {2021}
}

@inproceedings{demaine2011remarks,
  title        = {Remarks on separating words},
  author       = {Demaine, Erik D and Eisenstat, Sarah and Shallit, Jeffrey and Wilson, David A},
  booktitle    = {International Workshop on Descriptional Complexity of Formal Systems},
  pages        = {147--157},
  year         = {2011},
  organization = {Springer}
}

@article{SHALLIT199610,
  title   = {Automaticity I: Properties of a Measure of Descriptional Complexity},
  journal = {Journal of Computer and System Sciences},
  volume  = {53},
  number  = {1},
  pages   = {10-25},
  year    = {1996},
  issn    = {0022-0000},
  doi     = {https://doi.org/10.1006/jcss.1996.0046},
  url     = {https://www.sciencedirect.com/science/article/pii/S002200009690046X},
  author  = {Jeffrey Shallit and Yuri Breitbart}
}

@article{robson1989separating,
  title     = {Separating strings with small automata},
  author    = {Robson, John M},
  journal   = {Information processing letters},
  volume    = {30},
  number    = {4},
  pages     = {209--214},
  year      = {1989},
  publisher = {Elsevier}
}

@inproceedings{goralcik1986discerning,
  author    = {Goral{\v{c}}{\'i}k, P.
               and Koubek, V.},
  editor    = {Kott, Laurent},
  title     = {On discerning words by automata},
  booktitle = {Automata, Languages and Programming},
  year      = {1986},
  publisher = {Springer Berlin Heidelberg},
  address   = {Berlin, Heidelberg},
  pages     = {116--122},
  isbn      = {978-3-540-39859-2}
}

@inproceedings{yakaryilmaz2015discerning,
  author    = {Yakaryilmaz, Abuzer and Montoya, J. Andres},
  booktitle = {2015 Latin American Computing Conference (CLEI)},
  title     = {On discerning strings with finite automata},
  year      = {2015},
  pages     = {1-5},
  doi       = {10.1109/CLEI.2015.7360021}
}

@article{tran2023separating,
  title   = {Separating Words from Every Start State with Horner Automata},
  author  = {Tran, Nicholas},
  journal = {arXiv preprint arXiv:2309.02766},
  year    = {2023}
}

@inproceedings{tran2022variations,
  title        = {Variations of the Separating Words Problem},
  author       = {Tran, Nicholas},
  booktitle    = {International Conference on Implementation and Application of Automata},
  pages        = {165--176},
  year         = {2022},
  organization = {Springer}
}

@article{belovs2016can,
  title   = {Can one quantum bit separate any pair of words with zero-error?},
  author  = {Belovs, Aleksandrs and Montoya, Juan Andres and Yakary{\i}lmaz, Abuzer},
  journal = {arXiv preprint arXiv:1602.07967},
  year    = {2016}
}

@article{belovs2017conjecture,
  title     = {On a conjecture by Christian Choffrut},
  author    = {Belovs, Aleksandrs and Montoya, J Andres and Yakary{\i}lmaz, Abuzer},
  journal   = {International Journal of Foundations of Computer Science},
  volume    = {28},
  number    = {05},
  pages     = {483--501},
  year      = {2017},
  publisher = {World Scientific}
}

@article{robson1996separating,
  title   = {Separating words with machines and groups},
  author  = {Robson, John Michael},
  journal = {RAIRO-Theoretical Informatics and Applications-Informatique Th{\'e}orique et Applications},
  volume  = {30},
  number  = {1},
  pages   = {81--86},
  year    = {1996}
}

@article{currie1999separating,
  title   = {Separating words with small grammars},
  author  = {Currie, James D. and Petersen, Holger and Robson, John Michael and Shallit, Jeffrey},
  journal = {Journal of Automata, Languages and Combinatorics},
  volume  = {4},
  number  = {2},
  pages   = {101--110},
  year    = {1999}
}

@inproceedings{kuntewar2023separating,
  author    = {Kuntewar, Neha
               and Anoop, S. K. M.
               and Sarma, Jayalal},
  editor    = {Bordihn, Henning
               and Tran, Nicholas
               and Vaszil, Gy{\"o}rgy},
  title     = {Separating Words Problem over Groups},
  booktitle = {Descriptional Complexity of Formal Systems},
  year      = {2023},
  publisher = {Springer Nature Switzerland},
  address   = {Cham},
  pages     = {109--120},
  isbn      = {978-3-031-34326-1}
}

@book{upfal2005probability,
  title     = {Probability and Computing},
  author    = {Eli Upfal and Michael Mitzenmacher},
  year      = {2005},
  publisher = {Cambridge University Press}
}
\appendix

\section{Chernoff's Bound}
\begin{theorem}[{Chernoff's Bound~\cite[Theorem 4.4]{upfal2005probability}}]\label{thm:chb}
  Let $X_i, i = 1,\ldots m$ be i.i.d. Bernoulli random variables with parameter $p$, and let $\mu = m \cdot p$.
  There, for any $R \ge 6 \mu$, we have:
  \[\Prob{\sum_i X_i > R} \le 2^{-R}.\]
\end{theorem}

\end{document}